\documentclass[lettersize,journal]{IEEEtran}

\usepackage{amsmath,amsfonts,amsthm}
\usepackage{algorithmic}
\usepackage{algorithm}
\usepackage{array}
\usepackage[caption=false,font=normalsize,labelfont=sf,textfont=sf]{subfig}
\usepackage{textcomp}
\usepackage{stfloats}
\usepackage{url}
\usepackage{verbatim}
\usepackage{graphicx}
\usepackage{cite}
\newtheorem{theorem}{Theorem}

\newcommand{\ve}[1]{\mathbf{#1}}           % for vectors
\newcommand{\m}[1]{\mathbf{#1}}               % for matrices
\newcommand{\tr}[1]{{#1}^{\mkern-1.5mu\mathsf{T}}}              % for transpose
\newcommand{\conj}[1]{{#1}^{\ast}}
\newcommand{\norm}[1]{||{#1}||}              % norm
\newcommand{\frob}[1]{\norm{#1}_F}
\newcommand{\abs}[1]{\lvert{#1}\rvert}              % norm

\newcommand*{\trace}{\operatorname{trace}}

\newcommand*{\diag}{\operatorname{diag}}

\newcommand{\widebar}[1]{\overline{#1}}

\newcommand{\field}[1]{\mathbb{#1}}

\newcommand{\Complex}{\field{C}}

\begin{document}

\title{Optimal Structured Matrix Approximation for Robustness to Incomplete Biosequence Data}

\author{Chris Salahub, Jeffrey Uhlmann
        % <-this % stops a space
\thanks{Chris Salahub is a PhD student in Statistics at the University of Waterloo, Jeffrey Uhlmann is an Associate Professor in Electrical Engineering and Computer Science at the University of Missouri.}}% <-this % stops a space

\maketitle

\begin{abstract}
 We propose a general method for optimally approximating an arbitrary matrix $\m{M}$ by a structured matrix $\m{T}$ (circulant, Toeplitz/Hankel, etc.) and examine its use for estimating the spectra of genomic linkage disequilibrium matrices. This application is prototypical of a variety of genomic and proteomic problems that demand robustness to incomplete biosequence information. We perform a simulation study and corroborative test of our method using real genomic data from the Mouse Genome Database \cite{bultetal2019mouse}. The results confirm the predicted utility of the method and provide strong evidence of its potential value to a wide range of bioinformatics applications. Our optimal general matrix approximation method is expected to be of independent interest to an even broader range of applications in applied mathematics and engineering.
\end{abstract}

\begin{IEEEkeywords}
Structured matrix, circulant matrix, Toeplitz matrix, Frobenius norm, optimal approximation, genetics, linkage diseqbuilibrium.
\end{IEEEkeywords}

\section{Introduction}

\IEEEPARstart{L}{arge} matrices are ubiquitous to Big Data applications, especially in bioinformatics, to encode structural relationships among components of complex systems. These relationships may be physical in the case of connectivities among molecular components of genomic or proteomic sequences \cite{bai2010spectral,proteomics,natureProteo}, or they may represent statistical relationships encoded in covariance matrices \cite{covZhou,biometrika,baik2006eigenvalues,li2012two}. Matrix encodings permeate virtually all areas of science and engineering because they admit mathematical operations that can reveal subtle but important properties of a system that otherwise would be practically impossible to discern. Unfortunately, the size of matrices that arise in many bioinformatics applications are extremely large, ranging from $n\times n$ matrices with $n$ in the thousands up to tens of millions \cite{7440788}. This introduces enormous challenges both because virtually all nontrivial matrix operations have $O(n^3)$ complexity\footnote{Theoretical algorithms \cite{STRASSEN,CoppersmithW90} exist that reduce the exponent from 3 to as low as 2.36 \cite{JV21}, but the improved theoretical complexity does not translate to improved practical efficiency even for the largest matrices encountered in current applications \cite{galactic}.} and uncertainty in data collection frequently makes these matrices incomplete. 

In many cases, the computational and statistical costs associated with these large matrices are prohibitive unless practical approximations can be found. The most natural and common approach is to simply perform analysis on a small subset of the data, e.g., by global downsampling \cite{downsample}. This approach may be effective in applications involving statistics that will be preserved, but at a cost of coarser resolution in the form of larger error variances. However, subsampling cannot be performed for data points corresponding to critical features that thus cannot be entirely discarded. In this case, the only option may be to transform the dataset to one of lower complexity that retains approximate, though degraded, information about all features \cite{sparsification}. An example would be reduction of a large transportation network by removing roads for which alternative routes are known to exist so that effective route planning may be performed, though with potentially suboptimal results. 

In this paper we describe a general mechanism for approximating an arbitrary $n\times n$ matrix, i.e., having $O(n^2)$ parametric dimensionality, with a structured $n\times n$ matrix which is defined by as few as $O(n)$ parameters. In certain cases, this comes with enormous computational and memory savings for most matrix operations, e.g., spectral decompositions are reduced from $O(n^3)$ to $\tilde{O}(n)$, where $\tilde{O}$ hides logarithmic and polylogarithmic factors. For example, all common matrix operations performed on a circulant matrix have $O(n\log(n))=\tilde{O}(n)$ time complexity \cite{gray2006toeplitz}, as opposed to $O(n^3)$ for a general matrix. The first key question is whether optimal structured approximations for arbitrary matrices can be constructed efficiently. The second key question is whether the resulting approximations retain properties relevant to practical applications. The first question is answered in the following section, and the second question is subsequently addressed in Section \ref{sec:examples} with a study involving genomic analysis using data from the Mouse Genome Database \cite{bultetal2019mouse}.

\section{Structured Matrix Approximation}

Suppose we would like to approximate the $n \times n$ matrix
\begin{equation} \label{eq:Mdefn}
  \m{M} = \begin{bmatrix}
    m_{00} & m_{01} & \dots & m_{0,n-1} \\
    m_{10} & m_{11} & \dots & m_{1,n-1} \\
    \vdots & \vdots & \ddots & \vdots \\
    m_{n-1,0} & m_{n-1,1} & \dots & m_{n-1,n-1}
  \end{bmatrix}
\end{equation}
$m_{ij} \in \Complex$ with a structured matrix $\m{T} \in \Complex^{n \times n}$ for computational or analytical reasons, e.g., circulant $\m{T}$ for preconditioning \cite{chan1988optimal, venkatapathi2021circulant} and Toeplitz-Hankel $\m{T}$ for physical modelling \cite{narayanshastry2021toeplitz}. For any application, we prefer the approximating matrix $\m{T}$ to be optimal by some measure, commonly the Frobenius norm of the difference $\m{T} - \m{M}$ defined as
 \IEEEpubidadjcol

\begin{equation} \label{eq:frobnorm}
  \frob{\m{T} - \m{M}} = \sqrt{\trace \left ( \conj{(\m{T} - \m{M})}(\m{T} - \m{M}) \right )},
\end{equation}
where $\conj{\m{A}}$ is the conjugate matrix of $\m{A} \in \Complex^{n \times n}$. Minimizing this quantity for a circulant $\m{T}$ was the express goal of \cite{chan1988optimal} and was noted as a positive feature of the approximation in \cite{venkatapathi2021circulant}. Here we present a general result that can construct an optimal structured matrix approximation to an arbitrary given general matrix. 

\subsection{Structured matrices} \label{sec:struc}

We say $\m{T}  \in \Complex^{n \times n}$ is a structured matrix if it has entries $t_{ij}$, $0 \leq i,j \leq n-1$, following a pattern in $i$ and $j$, that is if
\begin{equation} \label{eq:generalStruc}
  t_{ij} = t_{f(i,j)}
\end{equation}
where $f:\{0,1, \dots, n-1\}^2 \mapsto \{0, 1, 2, \dots, K\}$ is the index function defining the membership of the index pair $i,j$ to an index set with a constant value. This implies a $k^{\text{th}}$ index set $\mathcal{T}_k = \{(i,j) | f(i,j) = k\}$ with cardinality $\lvert \mathcal{T}_k \rvert = n_k > 0$.

The index function $f(\cdot,\cdot)$ defines the structure of $\m{T}$ by defining the $\mathcal{T}_k$. Common structures and corresponding index functions are shown in Table \ref{tab:indexfuns}, though these functions are not unique for a given structure. Many candidate functions define identical index sets, for example Hankel matrices can take either $f(i,j) = j + i$ or $f(i,j) = 2(n-1) - j - i$.

\begin{table}[!h]
  \caption{Some common examples of structured index functions.} \label{tab:indexfuns}
  \begin{center}
  \begin{tabular} {|c|c|} \hline
    Structure & $f(i,j)$ \\ \hline
    Circulant & $(i - j) \bmod n$ \\
    Toeplitz & $j - i + n$ \\
    Hankel & $i + j$ \\ \hline
  \end{tabular}
  \end{center}
\end{table}

\subsection{Optimizing the Frobenius norm} \label{sec:approx}

Using the notation defined above, we obtain the following theorem.

\begin{theorem}[Means minimize $\frob{\m{T} - \m{M}}$] \label{thm:genOptimal}
  The Frobenius optimal approximating structured matrix $\m{T}$ with index function $f(i,j)$ for $\m{M}$ is given by $\m{T}_M$ with
  \begin{equation} \label{eq:meanStrucMat}
    t_{ij} = t_{f(i,j)} = \widebar{m}_{f(i,j)}
  \end{equation}
  where
  \begin{equation} \label{eq:strucMeans}
    \widebar{m}_k := \frac{1}{n_k} \sum_{\mathcal{T}_k} m_{ij}.
  \end{equation}
  is the mean of entries in $\m{M}$ over the corresponding index set. Furthermore, $\frac{1}{\sqrt{n}}\frob{\m{T}_M - \m{M}}$ is the total within-group standard deviation of entries in $\m{M}$ over all index sets.
\end{theorem}
\begin{proof}
   Take $\widebar{m}_k$ to be the mean of entries in $\m{M}$ for the $k^{\text{th}}$ index set as in Equation \ref{eq:strucMeans}, define the vector of all such means
  \begin{equation*}
    \widebar{ \ve{m} } = \tr{ (\widebar{m}_0, \widebar{m}_1, \dots, \widebar{m}_K) }.
  \end{equation*}
  Further, denote the vector of unique $t_k$ as
  \begin{equation*}
    \ve{t} = \tr{ ( t_0, t_1, \dots, t_K ) }
  \end{equation*}
  and the diagonal matrix of $n_k$ as
  \begin{equation*}
    \m{N} = \diag (n_0, n_1, \dots, n_K).
  \end{equation*}
  As Equation \ref{eq:frobnorm} is always positive, any $\m{T}$ which minimizes $\frob{\m{T} - \m{M}}$ will also minimize $\frob{\m{T} - \m{M}}^2$. Expanding gives
  \begin{align} 
    \trace \left ( \conj{(\m{T} - \m{M})}(\m{T} - \m{M}) \right ) =  & \trace \conj{\m{M}} \m{M} \nonumber \\
                                                                     & - \trace \conj{\m{M}} \m{T} \nonumber  \\
                                                                     & - \trace \conj{\m{T}} \m{M} \nonumber  \\
    & + \trace \conj{\m{T}} \m{T}. \label{eq:ToBeOptimized}
  \end{align}
  $\conj{\m{M}} \m{M}$ is constant in $\m{T}$, so can be ignored. The latter three terms can be considered individually to give $\trace \conj{\m{T}} \m{T}= \sum_{k = 0}^{K} n_k \conj{t}_k t_k$, $\trace \conj{\m{M}} \m{T} = \sum_{k = 0}^{K} n_k t_k \conj{ \widebar{m} }_k$, and $\trace \conj{\m{T}} \m{M} = \sum_{k = 0}^{K} n_k \conj{t}_k \widebar{m}_k.$
  So we seek to minimize
  \begin{equation*}
    F(\ve{t}) = \sum_{k = 0}^{K} n_k \conj{t}_k t_k - \sum_{k = 0}^{K} n_k \conj{t}_k \widebar{m}_k - \sum_{k = 0}^{K} n_k t_k \conj{ \widebar{m} }_k,
  \end{equation*}
  which we can write in matrix form as
  \begin{align}
    F(\ve{t}) & = \conj{ \ve{t} } \m{N} \ve{t} - \conj{ \ve{t} } \m{N} \widebar{ \ve{m} } - \conj{ \widebar { \ve{m} } } \m{N} \ve{t} \nonumber \\
    & = \conj{ \left ( \ve{t} - \widebar{ \ve{m} } \right ) } \m{N}  \left ( \ve{t} - \widebar{ \ve{m} } \right ) - \conj{ \widebar{ \ve{m} } } \m{N} \widebar{ \ve{m} }. \nonumber
  \end{align}
  As $n_k > 0$ for all $k = 0, 1, \dots, K$, $\m{N}$ is positive definite, and so the quadratic form $\conj{ \ve{x} } \m{N} \ve{x}$ has a minimum of zero when $\ve{x} = \ve{0}$. Therefore $F(\ve{t})$ is minimized for $\ve{t} = \widebar{ \ve{ m } }$ and has a minimum of
  \begin{equation} \label{eq:Fmin}
    F(\widebar{ \ve{m} }) = - \conj{ \widebar{ \ve{m} } } \m{N} \widebar{ \ve{m} } = - \sum_{k = 0}^K n_k \norm{\widebar{m}_k}^2.
  \end{equation}
  So $\m{T}_M$ is the Frobenius-optimal structured matrix $\m{T}$ approximating $\m{M}$. The residual $\m{T}_M - \m{M}$ has a squared Frobenius norm of
  \begin{align}
    \frob{\m{T}_M - \m{M}}^2 & = \sum_{k = 0}^K n_k \left ( \sum_{\mathcal{T}_k} \frac{\norm{m_{ij}}^2}{n_k} - \norm{\widebar{m}_k}^2  \right ) \nonumber \\
    & = \sum_{k = 0}^K n_k \sigma_k^2 \label{eq:TMmin}
  \end{align}
  where $\sigma_k^2 = \frac{1}{n_k} \sum_{\mathcal{T}_k} \left ( m_{ij} - \widebar{m}_k \right ) ^2$ is the variance of the $m_{ij}$ in the index set $\mathcal{T}_k$. Therefore we have
  \begin{equation*}
    \frac{1}{\sqrt{n}}\frob{\m{T}_M - \m{M}} = \sqrt{ \sum_{k = 0}^K \frac{n_k}{n} \sigma_k^2} ,
  \end{equation*}
  which is the total within-group standard deviation from the structured means.
\end{proof}

The above proof applies to any matrix with structured equalities, and it implies a reduction in its space complexity from $O(n^2)$ to $O(\abs{\ve{t}})$. In the particular case of a circulant matrix, this reduction achieves an $O(n)$ representation, and the optimal structured matrix can be computed using the discrete Fourier transform.

\subsection{Circulant matrices}

Circulant matrices, denoted $\m{C}$, have an index function
\begin{equation} \label{eq:circdefn}
  f(i,j) = (i - j) \bmod n
\end{equation}
and so contain $n$ unique values denoted $t_0, t_1, \dots, t_{n-1} \in \Complex$. They see widespread use in signal processing, computation, and physical modelling both due to their close relationship with the Fourier transform and their known eigensystem \cite{chan1988optimal,gray2006toeplitz,narayanshastry2021toeplitz}. $\m{C}$ has eigenvalues $\lambda_k$ for $k = 0, 1, \dots, n-1$ given by $\lambda_k = t_0 + \sum_{l = 1}^{n-1} t_l \omega^{lk}$ and a corresponding $k^{\text{th}}$ eigenvector given by $\ve{x}_k = \tr{(1, \omega^k, \omega^{2k}, \dots, \omega^{(M-1)k} )}$ where $\omega = \exp (\frac{2 \pi i}{n})$ is the complex $n^{\text{th}}$ root of unity and $i = \sqrt{-1}$.

Much of the utility of circulant matrices arises from this eigensystem. The $n \times n$ matrix of eigenvectors of $\m{C}$ scaled to be unitary,
\begin{equation} \label{eq:DFT}
  \m{F} = \frac{1}{\sqrt{n}} [\ve{x}_0 | \ve{x}_1 | \ve{x}_2 | \dots | \ve{x}_{n-1} ] = \tr{\m{F}},
\end{equation}
is simply the discrete Fourier transform (DFT), which provides an alternate route to compute $\m{T}_M$ in the circulant case. Consider the simple approximation algorithm in \ref{algo:circ}.

\begin{algorithm}[H]
\caption{Optimal circulant approximation.}\label{algo:circ}
\begin{algorithmic}
\STATE 
\STATE {\textsc{OPTCIRC}}$(\mathbf{M})$
\STATE \hspace{0.5cm}$ \mathbf{D} \gets \diag \left ( \mathbf{F} \mathbf{M} \conj{\mathbf{F}} \right ) $
\STATE \hspace{0.5cm}$ \m{C}_D = \conj{\m{F}} \m{D} \m{F} $
\STATE \hspace{0.5cm}\textbf{return} $ \mathbf{C}_D $
\end{algorithmic}
\label{alg1}
\end{algorithm}

As $\m{D}$ is diagonal and $\m{F}$ is the matrix of circulant eigenvectors, $\m{C}_{D}$ is circulant with eigenvalues $d_{jj}$, the diagonal values of $\m{D}$. To determine the elements $(\m{C}_{D})_{ij}$ in terms of $\omega$, $\ve{x}$, and $\m{M}$, first note
\begin{equation} \label{eq:Dexplicit}
  d_{jj} = \frac{1}{n} \sum_{k = 1}^n \sum_{l = 1}^n \omega^{(k - l)(j - 1)} m_{lk}.
\end{equation}
Analogously, taking $\conj{\m{F}} \m{D} \m{F}$ gives an $i,j$ element
\begin{eqnarray}
  (\m{C}_D)_{ij} & = & \conj{\m{F}} \m{D} \m{F} \nonumber \\
                 & = & \frac{1}{n^2} \sum_{l = 1}^n \sum_{k = 1}^n m_{lk} \conj{\ve{x}}_{k-l} \ve{x}_{i - j} \nonumber \\
                 & = & \frac{1}{n}  \sum_{l = 1}^n \sum_{k = 1}^n m_{lk} I \big ( i - j \equiv k - l \bmod n \big ) \label{eq:diagonalMeans}
\end{eqnarray}
where $I(A)$ is the indicator function which returns $1$ if $A$ is true and $0$ if $A$ is false.
As $n_k = n$ for all index sets in the circulant case, Equation \ref{eq:diagonalMeans} indicates that $\m{C}_D$ is generated by replacing the values of $M$ along each circulant diagonal by the corresponding diagonal mean. Therefore $\m{C}_D = \m{T}_M$ when $\m{T}$ is restricted to be circulant and so by Theorem \ref{thm:genOptimal} $\m{C}_D$ is Frobenius optimal.

Though related results have been noted for the case where $\m{M}$ is Toeplitz by \cite{chan1988optimal} and more generally in \cite{venkatapathi2021circulant}, both of these are particular examples of the more general result of Theorem \ref{thm:genOptimal}. Specifically, Algorithm \ref{algo:circ} makes no assumptions about $\m{M}$ and can be applied to any matrix.

\section{Application to Genetic Linkage Disequilibrium} \label{sec:examples}

In modern genome-wide association studies (GWAS), assessing the impact of a region of the genome on a measured trait requires an adjustment for linkage disequilibrium (LD) \cite{LanderBotstein1989, uffelmannetal2021gwas}. Briefly, LD is the correlation between the inheritance of genetic regions, or markers, over successive generations.
For regions $i$ and $j$ on chromosomes $c_i$ and $c_j$, $LD(i,j)$ is frequently computed from an additive measure of genetic distance $d(i,j)$ (often centiMorgans or cM) using the equation
\begin{equation} \label{eq:zcorr}
 LD(i,j) = I(c_i =c_j) \gamma e^{-\frac{d(i,j)}{50}}
\end{equation}
where $\gamma$ is a constant determined by the population characteristics \cite{salahub2022structural}. As $LD(i,j)$ is constant in $d(i,j)$, this implies the LD matrix is a structured matrix for $c_i = c_j$, and zero otherwise.

In practice, LD matrices are used in multiple testing adjustment to account for correlations between individual tests of markers \cite{conneelyboehnke2007tests, hanetal2009rapid}. Occasionally, missing data complicate direct use of the LD matrix by forcing pairwise computation of the LD matrix, which is therefore not guaranteed to be positive definite and can lead to negative eigenvalues if used directly. Certain multiple testing methods assume strictly positive eigenvalues and so require some adjustment to address these negative eigenvalues, such as setting them to zero \cite{LiJi2005}.

We propose a preprocessing step to improve the robustness of pairwise LD matrices to missing data. Instead of using the pairwise LD matrix, we can leverage the known structure of linkage disequilibrium given by Equation \ref{eq:zcorr} and Theorem \ref{thm:genOptimal} to compute the optimal theoretical matrix bytaking appropriate means of the pairwise LD. This not only replaces the observed matrix with one which follows the prescribed theoretical structure, but as the operation is equivalent to taking structured means, it should reduce the impact of individual missing entries.

\subsection{Simulated data}

To test this proposal, a simple simulation study was performed. 100 synthetic populations of 100 individuals were generated measured at 20 markers with $d(i,j) = 15$ cM and $c_i = c_j = 1$ for each pair, a setting corresponding to a Toeplitz LD matrix. For each population, the complete data was used to compute the observed LD matrix and determine its eigenvalues. An increasing proportion of observations were then removed completely at random from the data to simulate different levels of data completeness, and at each proportion the pairwise LD matrix was computed and used to generate the nearest Toeplitz matrix. The minimum eigenvalue and sum of squared errors in the ordered eigenvalues from the complete data LD matrix were computed and recorded for both the pairwise LD matrix and the nearest Toeplitz matrix. Figure \ref{fig:testres} shows the results. 

\begin{figure}[!h]
  \begin{centering}
    \begin{tabular}{c}
      \includegraphics[scale = 0.75]{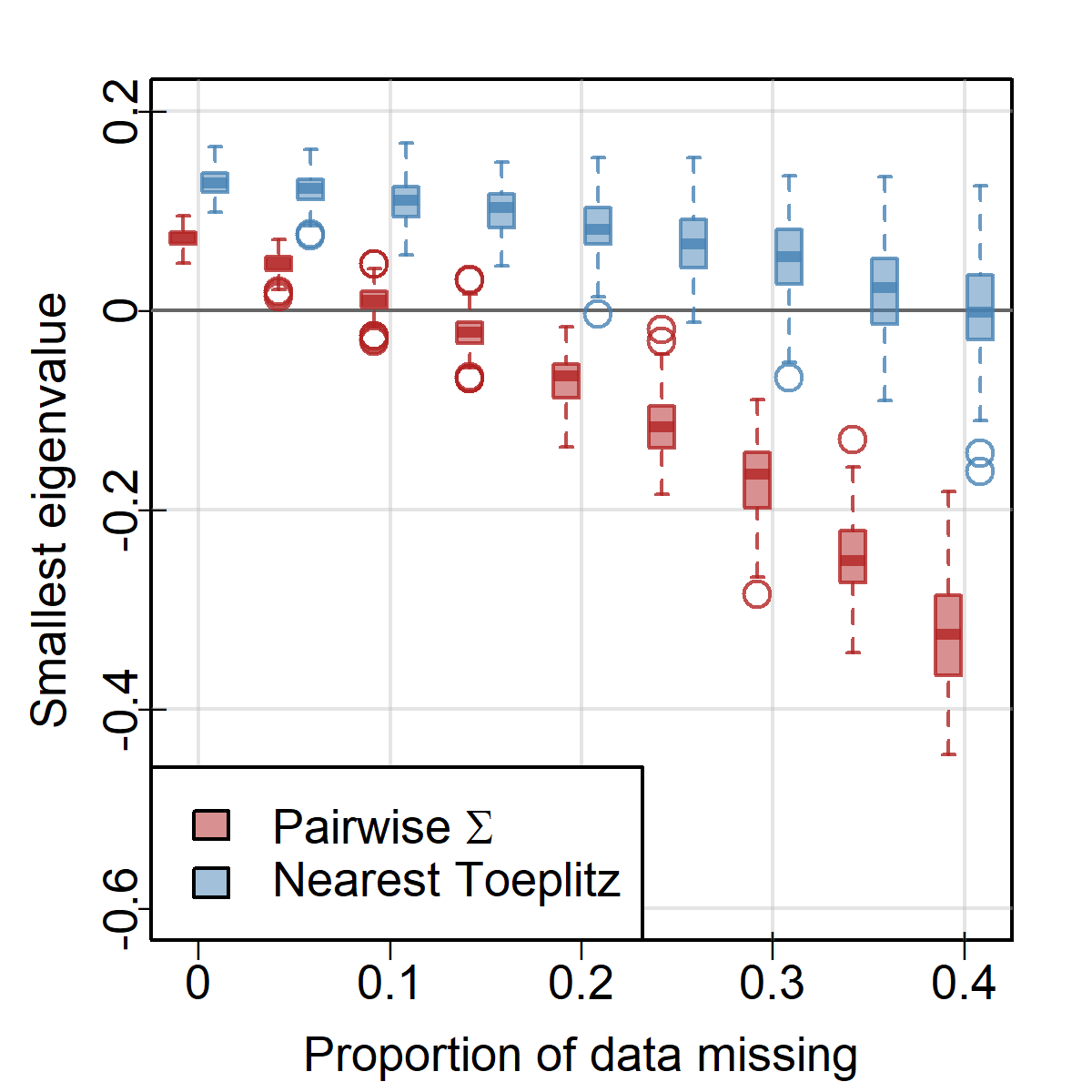} \\ 
      (a) \\
      \includegraphics[scale=0.75]{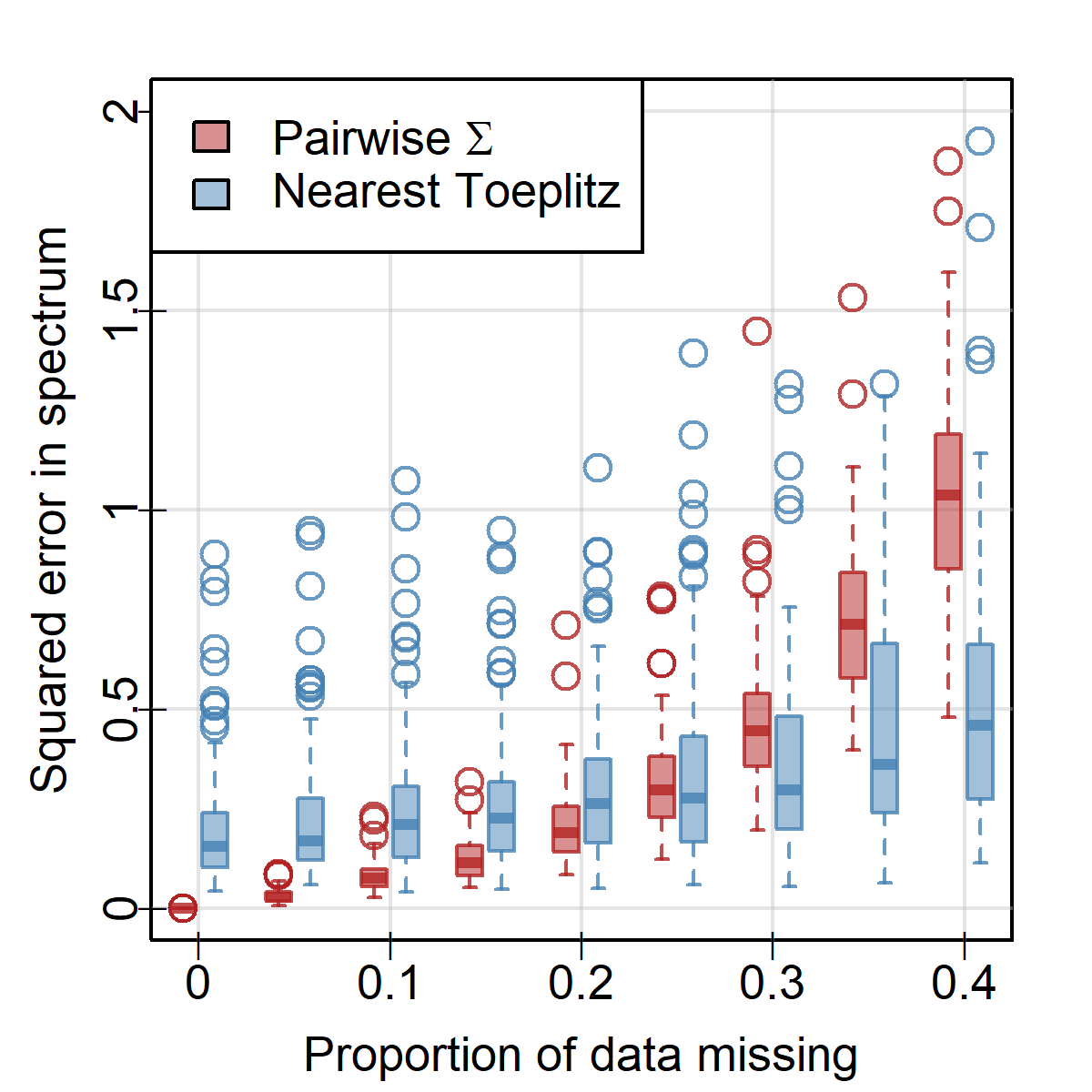} \\
      (b) 
    \end{tabular}
    \caption{Paired boxplots of the (a) minimum eigenvalues and (b) sum of squared errors in the ordered eigenvalues for the pairwise LD matrix and the nearest Toeplitz matrix by the proportion of data missing. The nearest Toeplitz, displayed to the right of the line for each pair of boxplots, is more robust to missing data than the pairwise LD matrix, displayed to the left of the line for each pair, but is biased when the data are complete.} \label{fig:testres}
  \end{centering}
\end{figure}

Figure \ref{fig:testres}(a) shows that the nearest Toeplitz is more robust to negative eigenvalues than the pairwise LD matrix alone. A vast majority of the 100 simulated populations have nearest Toeplitz matrices with no negative eigenvalues until more than a third of the data is missing, and even then only about 25\% produce negative eigenvalues. In contrast, the pairwise LD matrix has negative eigenvalues more than 75\% of the time when as little as 15\% of the data is missing. Figure \ref{fig:testres}(b) indicates this robustness also extends to the sum of squared differences between ordered eigenvalues, which does not depend greatly on the completeness of the data for the nearest Toeplitz but does for the pairwise LD. The cost of this robustness and the reduction in space complexity from $O(n^2)$ to $O(2n)$ is potential bias when the data is (mostly) complete.

\subsection{Real data}

Following the simulated study of the previous section, we now corroborate our findings by replicating the experiment with real data. The Mouse Genome Database (MGD) is a public repository of information on mouse genomics, including references which describe hundreds of thousands of markers and their measured values in several experiments on live mice \cite{bultetal2019mouse}. One of these data sets records markers measured on the backcross of a mutated strain of \emph{Mus musculus} and a strain of \emph{Mus spretus} in \cite{roweetal1994jaxbsb}. Partial measurement of 5951 markers across all chromosomes on 94 animals is reported, with complete observations present on 2624 of the markers used to create a 2624$\times$2624 observed correlation matrix. In contrast to the previous simulated example, these markers are not positioned uniformly across the chromosomes; the cM distances between adjacent markers differ greatly.

Instead of focusing on all of the markers, however, we make the real data comparable to the simulated case by considering only those markers on chromosome 1. This leaves 199 markers measured across the chromosome with adjacent distances ranging from less than 0.01 cM to 4.37 cM. The cM distances were used to generate a theoretical correlation matrix according to Equation \ref{eq:zcorr}, these theoretical correlations were rounded to two decimal places, and these rounded correlations were treated as the level sets of the theoretical structured matrix. Note that unlike the examples illustrated above, these level sets do not correspond with a named structured matrix. Nonetheless, Theorem \ref{thm:genOptimal} dictates that the optimal structured approximation in the Frobenius norm is given by means computed over the level sets. An unfortunate consequence of the lack of correspondence with a named matrix is that there is no guarantee that the result will be a valid correlation matrix.

After these steps, we have two correlation matrices: the observed matrix on the full data and the theoretical correlation matrix based on cM distances. To examine the impact of missing data on the computed eigenvalues, marker measurements were removed completely at random 100 times for each of a range of proportions of missingness and the pairwise LD matrix and the optimal structured approximation based on the theoretical level sets were generated. For both the pairwise LD matrix and the optimal structured matrix, the minimum eigenvalue and sum of squared errors from the matrix of complete observations was computed. Figure \ref{fig:mgi} displays the results.

\begin{figure}[!h]
  \begin{centering}
    \begin{tabular}{c}
      \includegraphics[scale = 0.75]{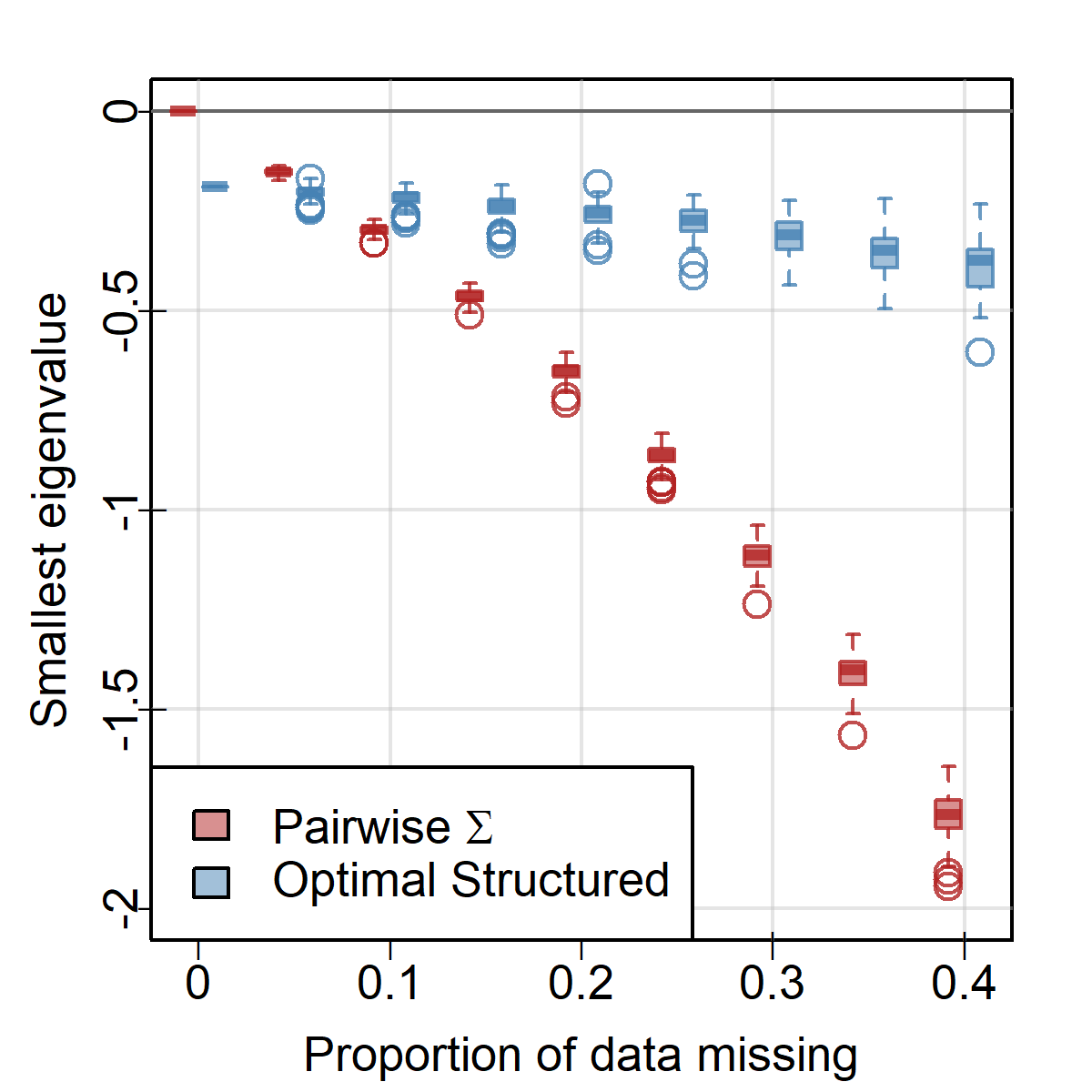} \\ 
      (a) \\
      \includegraphics[scale=0.75]{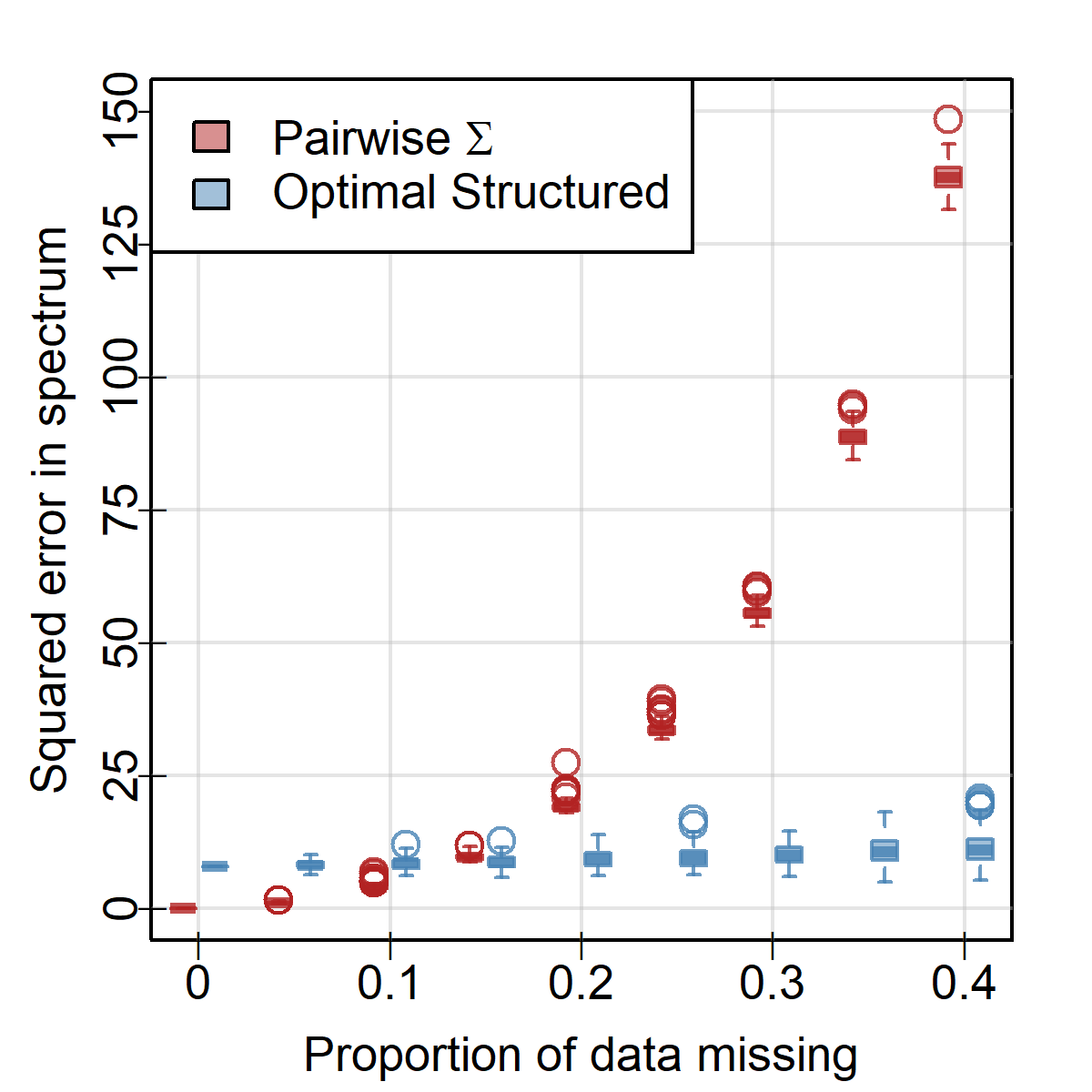} \\
      (b) 
    \end{tabular}
    \caption{Paired boxplots of the (a) minimum eigenvalues and (b) sum of squared errors in the ordered eigenvalues for the pairwise LD matrix (to the left of the corresponding line) and the optimal structured matrix (to the right of the corresponding line) by the proportion of data missing. The bias from approximation is more serious in this case than the simulated example: negative eigenvalues are produced for the complete data.} \label{fig:mgi}
  \end{centering}
\end{figure}

First, note that the larger correlation matrix for the real data example has resulted in less relative variability for our summaries of eigenvalues in both the pairwise LD and optimal structured matrices, making the difference in perfomance even starker. Just as in the simulated case, the optimal structured matrix proves far more robust to missing data. Both the minimum eigenvalue and the sum of squared errors barely change on average as the proportion of missing data increases. In contrast, the pairwise LD matrix has an error that grows in the proportion of missing data and a minimum eigenvalue which continues to decrease as data is removed. When as little as 15\% of the data is missing, the optimal structured matrix has a better distribution in both metrics.

The bias observed in the simulated case has more serious consequences in the real data, however. For mostly complete data, the optimal structured matrix produces negative eigenvalues, and so the optimal structured matrix clearly does not correspond with a correlation matrix. This limits its use for cases with nearly complete data, where an optimal structured approximation should only be used if the structure is Toeplitz, circulant, or of another class with known qualities to its eigenspectrum. Future work will examine a more sophisticated approximation in the form of a structured matrix plus a diagonal matrix. This form retains a reduced space complexity, but the diagonal components provide variables that can potentially be optimized to ensure desired properties, e.g., postive semidefiniteness for correlation matrices. 

\section{Conclusion}

In this paper, we have presented a general method for obtaining structured approximations to arbitrary matrices. Specifically, we have proved that the Frobenius norm of $\m{M}$ to a matrix $\m{T}$ with structural equality of certain entries is minimized by replacing the entries in $\m{M}$ by their means for each index set of $\m{T}$. The value of this minimum corresponds to the weighted standard deviation of entries within each index set, giving an intuitive way to measure the expected error for this approximation. In the circulant case, this is equivalent to taking the circulant matrix with eigenvalues $\diag \left ( \m{F} \m{M} \conj{\m{F}} \right )$, where $\m{F}$ is the DFT matrix.

We believe that our general structured approximation method is of significant independent theoretical and practical interest. Theoretically, it promises a simple and computationally-efficint algorithm to compute the Frobenius-optimal matrix with structured equalities to an arbitrary matrix. Evidence for its practical value comes from our example involving the creation of a structured approximation of the pairwise genetic linkage disequilibrium matrix that is highly robust to missing data. In both simulated and real data examples, the eigenvalues changed little for the optimal structured matrix when as much as 40\% of the data was missing, though at the cost of some bias when the data was nearly complete. Despite this remarkable result, it must be emphasized that the new structured approximation method is potentially applicable to any problem involving large-scale matrices for which standard linear algebraic methods cannot practically be applied.

\section{References}

\bibliographystyle{ieeetran}
\bibliography{./bibliography}

\end{document}